\documentclass[12pt,reqno,centertags]{amsart}
\usepackage{a4}
\usepackage{amsmath,amsthm,amscd,amssymb}
\usepackage{latexsym}
\usepackage{upref}
\usepackage{hyperref}
\usepackage{upgreek}
\usepackage{textgreek}
\usepackage{cite}
\usepackage{tikz}
\usepackage{enumerate}
\usepackage{multicol}
\usepackage{accents}
\usepackage{lmodern}
\usepackage{microtype}
\usepackage{graphicx}
\usepackage{subcaption}
\usepackage{dsfont}
\usepackage{colonequals}
\usepackage{bbm}
\usepackage{mathdots} 
\UseRawInputEncoding
\usetikzlibrary{arrows,positioning,automata}

\newcommand{\HH}{{\mathcal{H}\!\oplus\!\mathcal{H}}}
\newcommand\oP[1]{{#1}_{\tiny{\odot}}}
\newcommand\vE[2]{{\footnotesize\begin{pmatrix}{#1}\\[-.5mm]{#2}\end{pmatrix}}}

\newcommand{\abs}[1]{\left|#1\right|}
\newcommand\lrb[1]{\left\lbrace#1 \right\rbrace}
\newcommand\lrp[1]{{\left(#1 \right)}}
\newcommand\C{{\mathbb C}}
\newcommand\N{{\mathbb N}}
\newcommand\R{{\mathbb R}}
\newcommand\cc[1]{\overline {#1}}
\newcommand\ip[2]{\left\langle {#1},{#2} \right\rangle}
\newcommand\no[1]{\left\| {#1} \right\|}
\newcommand\kT[1]{\mathfrak K\!\left(#1\right)}
\newcommand\dS{{\rm\bf N}}
\DeclareMathOperator{\dom}{dom}
\DeclareMathOperator{\ran}{ran}
\DeclareMathOperator{\mul}{mul}
\DeclareMathOperator{\Span}{span}

\allowdisplaybreaks
\numberwithin{equation}{section}

\newtheorem{theorem}{Theorem}[section]
\newtheorem{lemma}[theorem]{Lemma}
\newtheorem{corollary}[theorem]{Corollary}

\newtheorem{proposition}[theorem]{Proposition}
\theoremstyle{definition}
\newtheorem{definition}[theorem]{Definition}

\theoremstyle{remark}
\newtheorem{remark}[theorem]{Remark}

\begin{document}
\title[The Krein transform and semi-bounded extensions of relations]{The Krein transform and semi-bounded extensions of semi-bounded linear relations}
\author[J. I. Rios-Cangas]{Josu\'e I. Rios-Cangas}
\address{Departamento de Matem\'aticas, Universidad Aut\'onoma Metropolitana, Iztapalapa Campus,  San Rafael Atlixco 186, 09340 Iztapalapa, Mexico City.}
\email{jottsmok@xanum.uam.mx}

\date{\today}
\subjclass{Primary 47A06; Secondary 47A45; 47B65}

\keywords{Linear relations; Krein transform; Semi-bounded extensions}

\begin{abstract}
The Krein transform is the real counterpart of the Cayley transform and gives a one-to-one correspondence between the positive relations and symmetric contractions. It is treated with a slight variation of the usual one, resulting in an involution for linear relations. On the other hand, a semi-bounded linear relation has closed semi-bounded symmetric extensions with semi-bounded selfadjoint extensions.  A self-consistent theory of semi-bounded symmetric extensions of semi-bounded linear relations is presented. By using The Krein transform, a formula of positive extensions of quasi-null relations is provided.
\end{abstract}

\maketitle

\section{Introduction}\label{s1}
For a separable Hilbert space $(\mathcal  H,\ip{\cdot}{\cdot})$ over $\C$, with inner product antilinear in its first entry, we consider linear relations \cite{MR0123188} (or multivalued linear operators \cite{MR1631548}) as  linear sets in the Hilbert space $\HH$, where  $\HH$ is the orthogonal sum of two copies of $\mathcal  H$ \cite[Chap.\,2,\,Sect.\,3.3]{MR1192782}. A semi-bounded linear relation is a symmetric relation $A$ for which  $A-mI\leq 0$ or $A-mI\geq 0$, for some $m\in\R$. It turns out to be that $A$ is either selfadjoint or has closed symmetric extensions with selfadjoint extensions. Particularly, $A$ has semi-bounded selfadjoint extensions (see  Proposition~\ref{prop:existence-sbe}). 

In this paper, we are interested in describing the semi-bounded extensions of semi-bounded relations and, in particular, the positive extensions of positive linear relations. In the first instance, one can determine all the symmetric extensions of a semi-bounded relation employing the so-called second von Neumann formula for linear relations (see for instance \cite[Thm.\,4.7]{MR4082306}), where the Cayley transform plays a crucial role. We will attempt to use analogous reasoning of the previously cited formula to obtain through the Krein transform an explicit formula that describes the positive extensions of a positive relation. It is worth mentioning that the theory of positive selfadjoint extensions has recently applications to Sturm-Liouville operators and elliptic second order differential operators such as Laplacian and the Aharonov-Bohm operators \cite{MR4410827,MR4013751}.   

Krein in his work \cite{MR0024574} proved that there exist two peculiar positive selfadjoint extensions of a densely defined positive symmetric operator $A$. Namely, the Friedrichs and Krein-von Neumann extensions of $A$ (cf.\cite{MR1902794} and \cite[Sect.\,13.3]{MR2953553}), which allow exhibiting  a characterization of all positive selfadjoint extensions of $A$. The last result was extended in \cite{MR500265} for nondensely defined symmetric operators and, actually, for symmetric linear relations. Besides, through the Friedrichs-Krein type extensions and with the help of closed semi-bounded forms, one can show a description of all the semi-bounded self-adjoint extensions of a semi-bounded relation \cite{MR3971207}. 

This paper provides a natural and intrinsic way to describe the semi-bounded extensions of semi-bounded linear relations. We emphasize that the extensions we consider here are without exit to a larger space \cite[Appendix\,1]{MR1255973}. In this fashion, and for the reader's convenience, we first deal in Section~\ref{s2} with a general framework on the theory of linear relations, which can be found in \cite{MR4082306,MR4091412}. We characterize in Section~\ref{s3} the semi-bounded relations concerning its quasi-regular set (see Theorem~\ref{eq:GS-bound} and Corollary~\ref{cor:charact-usb}), its spectrum, and its operator part (see Theorem~\ref{th:Alusb-oP}). To describe the semi-bounded extensions of semi-bounded relations, we address in Section~\ref{s4} the Krein transform for linear relations, which is the real counterpart of the Cayley transform. Here, we attend the Krein transform with a minor variation of the usual one \cite[Sect.\,13.4]{MR2953553}, which yields the additional ``desirable'' property of being an involution for linear relations (the first property of \eqref{eq:Kt-properties}). Besides, this transform gives a univocal correspondence between positive relations and symmetric contractions (see Theorem~\ref{th:positive-Kt-contraction}). After that, we show in Section~\ref{s5} that for a closed lower semi-bounded relation $A$ with finite index $\eta_A$ and greatest lower bound $m_A$, every number $\alpha\in(\infty,m_A)$ is in one-to-one correspondence with the lower semi-bounded selfadjoint extension of $A$, which has $\alpha$ as an eigenvalue of multiplicity $\eta_A$ and greatest lower bound. A similar situation occurs for upper semi-bounded relations (see Theorem~\ref{th:sa-extension-lusb}). Also, we prove in Theorem~\ref{th:positive-extensions} that any closed symmetric extension $S$ of a closed positive relation $A$, is positive if and only if it can be decomposed uniquely into an orthogonal sum $S=A\oplus L$, where $L$ is a positive relation contained in $A^*$. Additionally, Theorem~\ref{th:positive-extensions-qn} shows that the positive extension $S$ follows a similar structure to the second von Neumann formula \eqref{eq:von02d} when $A$ is quasi-null (see Definition~\ref{def:quasi-null}). 

As an illustration of the general results of this work, we provided in the last section the description of all the semi-bounded extensions (in particular, positive and quasi-null extensions), as well as their spectra, of the adjacency operator of the infinite star graph.

\section{Preliminaries about linear relations}\label{s2}
A linear set $T$ of $\HH$ is called \emph{linear relation} (or simply relation), where
\begin{align*}
\dom T&\colonequals\lrb{f\in \mathcal  H\,:\, \vE fg\in T}\,,&
\ran T&\colonequals\lrb{g\in \mathcal  H\,:\, \vE fg\in T}\,,\\[1mm]
\ker T&\colonequals\lrb{f\in \mathcal  H\,:\, \vE f0\in T}\,,&
\mul T&\colonequals\lrb{g\in \mathcal  H\,:\, \vE 0g\in T}\,,
\end{align*} 
denote the domain, range, kernel and multivalued of $T$, respectively.

We emphasize the term ``linear set'' instead of the usual ``subspace'', since the last one in this work is reserved for closed linear sets. Thereby, the \emph{closure} $\cc T$ of a relation $T$ is a subspace of $\HH$. 

For relations $T,S$ and $\alpha\in\C$, we consider the following relations
\begin{align*}
T+S&\colonequals\lrb{\vE f{g+h}\,:\,\vE fg\in T,\,\vE fh\in S}\,,&
\zeta T&\colonequals \lrb{ \vE f{\zeta g}\,:\,\vE fg\in T}\,,\\[1mm]
ST&\colonequals\lrb{ \vE fk\,:\,\vE fg\in T,\,\vE gk\in S}\,,& T^{-1}&\colonequals\lrb{\vE gf\,:\, \vE fg\in T }\,.
\end{align*}
We also deal with the \emph{adjoint} of $T$ given by
\begin{gather*}
 T^*\colonequals\lrb{\vE hk\in \HH\,:\,\ip kf=\ip hg,\, \forall\, \vE fg\in T}\,,
\end{gather*}
which is a closed relation such that
\begin{align}\label{eq:p-Adjoint}
T^*&=(-T^{-1})^{\perp},&S\subset  T&\Rightarrow T^*\subset S^*,\nonumber\\
T^{**}&=\cc T,& (\alpha T)^*&=\cc{\alpha} T^*,\,\mbox{ with } \alpha\neq0,\\
(T^*)^{-1}&=(T^{-1})^*,&\ker  T^*&=(\ran  T)^{\perp}.\nonumber
\end{align}
 The last property of \eqref{eq:p-Adjoint} decomposes the Hilbert space as follows:
 \begin{gather}\label{eq:Hdecomposition}
 \mathcal  H=\cc{\ran T}\oplus\ker T^*\,.
 \end{gather} 
 
 \begin{remark}\label{rm:dist-adjoint} If $\dom T\subset \dom S$ and $\dom (T+S)^*\subset \dom S^*$, then 
\begin{gather*}
(T+S)^*=T^*+S^*\,.
\end{gather*}
The above is adapted from the proof of \cite[Thm.\,3.41]{MR0123188}. 
\end{remark}
 
In the context of linear relations, the notion to be bounded is not unique \cite{MR1631548} (see also \cite{MR3750592,MR2993376}). 
In this work, a relation $T$ is \emph{bounded} if there exists $C>0$, such that $\no g\leq C\no f$, for all $\vE fg\in T$. In this sense, the multivalued of any bounded relation $T$ is trivial, viz. $T$ is a linear operator.  

\begin{remark} Following the proof of \cite[Thm.\,3.2.3]{MR1192782}, for closed relations $T, S$, with $S$ bounded, one readily computes that $T+S$ is closed. 
\end{remark}

The \emph{quasi-regular} set
$\hat\rho(T)$ of a relation $T$ is 
\begin{gather*}
\hat\rho(T)\colonequals\lrb{\zeta \in \C\
:\ (T-\zeta I)^{-1}\mbox{ is bounded}}\,,
\end{gather*} 
which is an open set in $\C$. Let $\mathcal {B(H)}$ denote the class of all bounded operators having the whole space $\mathcal  H$ as their domain and define the \emph{regular} set of $T$ by
\begin{gather*}
\rho(T)\colonequals\lrb{\zeta \in \C\ :\ (T-\zeta I)^{-1}\in\mathcal {B(H)}}\,,
\end{gather*}
It is not difficult to see that the regular set is an open subset of the quasi-regular set. Besides, it is convenient to take the regular set over closed relations, otherwise it is empty.

Let us consider the following spectral sets of a relation $T$.
\begin{align*}
\sigma(T)&\colonequals\C\backslash \rho(T), &\mbox{(spectrum)}\\
\hat\sigma(T)&\colonequals\C\backslash \hat\rho(T),&\mbox{(spectral core)}\\
\sigma_p(T)&\colonequals\{\zeta \in \C\, :\, \ker (T-\zeta I)\neq \{0\}\},&\mbox{(point spectrum)}\\
\sigma_d(T)&\colonequals\{\zeta \in \sigma_{p}(T)\, :\, \dim\ker (T-\zeta I)<\infty\},&\mbox{(discrete spectrum)}\\
\sigma_p^{\infty}(T)&\colonequals\{\zeta \in \sigma_{p}(T)\, :\, \dim\ker (T-\zeta I)=\infty\},&\mbox{(point non-discrete spectrum)}\\
\sigma_c(T)&\colonequals\{\zeta \in \C\, :\, \ran (T-\zeta I)\neq \overline{\ran (T-\zeta I)}\}.&\mbox{(continuous spectrum)}
\end{align*} 
The spectral core satisfies 
\begin{gather*}
\hat\sigma(T)=\sigma_p(T)\cup\sigma_c(T)\,.
\end{gather*}

We shall focus on the following decomposition, which will be practical in the sequel. For a closed relation $T$, we denote the \emph{multivalued part} and \emph{operator part} of $T$ by
\begin{gather}\label{eq:multivalued-operator.parts}
T_{\infty}\colonequals\lrb{\vE 0g\in T}\quad\mbox{and}\quad \oP T\colonequals T\ominus T_{\infty}\,,
\end{gather}
respectively.  The operators in \eqref{eq:multivalued-operator.parts} are closed and satisfy 
\begin{gather}\label{eq:decomposition-Tclosed}
T=\oP T\oplus T_\infty\,.
\end{gather}
The operator part $\oP T$ is a linear operator, while the multivalued part $T_\infty$ is a purely multivalued relation. The decomposition \eqref{eq:decomposition-Tclosed} allows studying spectral properties of $T$ by means of $\oP T$.

For two relations $T,S$ in $\HH$, we define $T_S$ as the relation in Hilbert space $\lrp{\mul S}^\perp\oplus\lrp{\mul S}^\perp$ (the orthogonal sum of two copies of $\lrp{\mul S}^\perp$, with inner product inherited from $\HH$) given by
\begin{gather}\label{eq:restrict-relation}
T_S\colonequals T\cap \lrp{\mul S}^\perp\oplus\lrp{\mul S}^\perp\,.
\end{gather}
The relations $T$ and $T_S$ are closed simultaneously. It is easy to verify that $(T^{-1})_T=(T_T)^{-1}$ and when $T$ is closed, $T_T=(\oP T)_T$, i.e., $T_T$ is a closed operator. Sometimes it is convenient to consider $T_S$ as a linear relation in $\HH$.

\begin{remark}\label{rmk:dom-multperp-spectral} If $T$ is a closed relation with $\dom T\subset\lrp{\mul T}^\perp$, then 
\begin{gather}\label{eq:Tt-ToP}
T_T=(\oP T)_T=\oP T\cap \lrp{\mul T}^\perp\oplus\lrp{\mul T}^\perp=\oP T\,,
\end{gather}
since $\ran \oP T\subset (\mul T)^\perp$. In addition, (cf.\cite[Thm.\,2.10]{MR4082306})
\begin{align}\label{eq:espectral-properties-inherited}
\begin{aligned}
\sigma(T)&=\sigma(T_T)\,,&\hat\sigma(T)&=\hat\sigma(T_T)\,,&\sigma_c(T)&=\sigma_c(T_T)\,,\\
\sigma_p(T)&=\sigma_p(T_T)\,,&\sigma_p^\infty(T)&=\sigma_p^\infty(T_T)\,,&\sigma_d(T)&=\sigma_d(T_T)\,.
\end{aligned}
\end{align}
\end{remark}

\section{Symmetric and semi-bounded linear relations}\label{s3}
We start this section with a brief discussion of symmetric relations.
\begin{definition}
A linear relation $A$ is said to be \emph{symmetric} if $A\subset A^*$ and \emph{selfadjoint} if $A=A^*$. A symmetric relation $A$ is \emph{maximal symmetric} if every symmetric extension $S$ of $A$ satisfies $S=A$.
\end{definition}
It is not difficult to show the following equivalences:
\begin{enumerate}[(i)]
\item $A$ is a symmetric relation.
\item $\ip fk=\ip gh$, for all $\vE fg,\vE hk\in A$.
\item $\ip fg\in\R$, for all $\vE fg\in A$.
\end{enumerate}

\begin{remark}\label{rmk:charc-symmetric-A}
There exists another characterization of symmetric relations in terms of the quasi-regular set. Namely, $A$ is symmetric if and only if $\C\backslash \R\subset \hat\rho(A)$ and
\begin{gather*}
\no{\lrp{A-\zeta I}^{-1}}\leq 1/\abs{{\rm Im\,} \zeta}\,,
\end{gather*}
for all $\zeta\in\C\backslash \R$ \cite[Rmk.\,3.2]{MR4082306}. In this fashion, $\hat\sigma(A)\subset\R$.
\end{remark}
 It is clear that a maximal symmetric relation is closed. Besides, a selfadjoint relation is maximal symmetric.
 \begin{remark}\label{rmk:Asym-AA}
For a symmetric relation $A$, it follows that $\dom A\subset(\mul A)^\perp$. Moreover, if $A$ is maximal symmetric then $\cc{\dom A}=(\mul A)^\perp$. The above is adapted from \cite[Lem.\,2.1]{MR3057107}. So, by virtue of Remark~\ref{rmk:dom-multperp-spectral}, if $A$ is a closed symmetric relation, then so is $A_A$, and $A$ satisfies the spectral properties \eqref{eq:espectral-properties-inherited}. 
 \end{remark}
 
It is convenient to consider the \emph{deficiency space} and \emph{deficiency index} (or index, for short) of a closed symmetric relation $A$, given by
 \begin{gather*}
 \dS_\zeta(A^*)\colonequals\lrb{\vE f{\zeta f}\in A^*}\quad\mbox{and}\quad\dim \dS_{\zeta}(A^*)\,, \quad \zeta\in\C
 \end{gather*}
 respectively. The index of $A$ is constant on each connected component of $\hat\rho(A)$. One has from Remark~\ref{rmk:charc-symmetric-A} that the connected components $C_+$ and $C_-$ (the upper and lower half-planes) are contained in $\hat\rho(A)$. So, the indices
\begin{gather}\label{eq:deficiency-index}
\dim \dS_{\zeta}(A^*)\,,\quad \dim \dS_{-\zeta}(A^*)\,,\quad \zeta\in\C_-
\end{gather}
remain constant over $\C_-$.

The following assertion can be found in \cite[Prop.\,4.1]{MR4091412}.
\begin{proposition}\label{prop:characterization-sa}
For  a closed symmetric relation $A$, the following are equivalent:
\begin{enumerate}[(i)]
\item $A$ is selfadjoint.
\item $\dim \dS_{\pm i}(A^*)=0$.
\item $\sigma(A)=\hat\sigma(A)\subset\R$.
\end{enumerate}
\end{proposition}

\begin{remark}\label{rm:Aselfadjoint-restriction}
As a consequence of Proposition~\ref{prop:characterization-sa} and Remark~\ref{rmk:Asym-AA}, if $A$ is a selfadjoint relation in $\HH$ then $A_A$ is a selfadjoint operator in $\lrp{\mul A}^\perp\oplus\lrp{\mul A}^\perp$. In this fashion, we may consider the spectral measure $E_{A}$ of $A_A$, on the Borel $\sigma$-algebra $\mathcal  B(\R)$. Hence, it follows from \eqref{eq:Tt-ToP} that
\begin{gather*}
A=A_A\oplus A_\infty=\int_{\R}\lambda dE_A(\lambda)\oplus A_\infty\,.
\end{gather*}
\end{remark}

In what follows, we deal with a particular class of symmetric relations known as semi-bounded relations (cf. \cite[Sect.\,3.1]{MR2953553} for operators). 
\begin{definition}\label{def:lsb-usb}
A relation $A$ is \emph{lower semi-bounded} (l.s.b. for short) if there exists $m\in\R$ such that
\begin{gather*}
\ip fg\geq m\no f^2\,,\quad \mbox{for all } \vE fg\in A\,.
\end{gather*}
A relation $A$ is \emph{upper semi-bounded} (briefly u.s.b.) if $-A$ is l.s.b. If $A$ is l.s.b. or u.s.b. then $A$ is called semi-bounded.
\end{definition}

It is clear that A semi-bounded relation is symmetric. Besides, every l.s.b. (resp. u.s.b.) relation $A$ has $m_A$ as the greatest lower bound (resp. $M_A$ as the smallest upper bound), which is given
by 
\begin{align}\label{eq:GS-bound}
\begin{split}
m_A&\colonequals \inf\lrb{\ip fg\,:\,\vE fg\in A\,,\no f=1}\\
\Big(\mbox{resp. }\, M_A&\colonequals \sup\lrb{\ip fg\,:\,\vE fg\in A\,,\no f=1}\Big)\,.
\end{split}
\end{align}
Moreover, the closure of a semi-bounded relation is semi-bounded of the same type with the same bound~\eqref{eq:GS-bound}. Note that $m_{-A}=-M_A$.

\begin{theorem}\label{th:charact-lsb}
A relation $A$ is lower semi-bounded if and only if it is symmetric, there exists $m_A\in\R$ such that $(-\infty,m_A)\subset \hat\rho(A)$ and for all $\alpha<m_A$,
\begin{gather}\label{eq:bound-Alsb}
\no{(A-\alpha I)^{-1}}\leq (m_A-\alpha)^{-1}\,.
\end{gather}
\end{theorem}
\begin{proof}
We first suppose that $A$ is l.s.b., then it is symmetric. Let $\alpha<m_A$, with $m_A$ as in \eqref{eq:GS-bound}, and $\vE hk\in(A-\alpha I)^{-1}$, i.e, $\vE k{h+\alpha k}\in A$. Since $A$ is l.s.b.,
\begin{gather*}
\ip kh+\alpha\no k^2=\ip k{h+\alpha k}\geq m_A \no k^2\,,
\end{gather*}
In this fashion, for $k\neq 0$ (otherwise, the following inequality is direct),
\begin{gather*}
\no k=\frac1{(m_A-\alpha)\no k}\ip kh\leq\frac1{m_A-\alpha}\no h\,,
\end{gather*}
whence it follows \eqref{eq:bound-Alsb} and $(-\infty,m_A)\subset \hat\rho(A)$. Conversely, consider a negative real number $\alpha<m_A$ and let $\vE fg\in A$. Then, $\vE{g-\alpha f}{f}\in (A-\alpha I)^{-1}$ and by \eqref{eq:bound-Alsb}, since $A$ is symmetric,
\begin{gather*}
\no f^2(m_A-\alpha)^2\leq \no{g-\alpha f}^2=\no g^2+\alpha^2\no f^2-2\alpha\ip fg\,,
\end{gather*}
wherefrom, 
\begin{gather*}
\no f^2\lrp{m_A-\frac{m_A^2}{2\alpha}}+\frac 1{2\alpha}\no g\leq\ip fg\,.
\end{gather*}
Thus, letting $\alpha$ tends to minus infinity, $\ip fg\geq m_A\no f^2$. Hence, $A$ is l.s.b.
\end{proof}
The following assertion is straightforward from Theorem~\ref{th:charact-lsb}.
\begin{corollary}\label{cor:charact-usb}
A relation $A$ is upper semi-bounded if and only if it is symmetric, there exists $M_A\in\R$ such that $(M_A,\infty) \subset \hat\rho(A)$ and for all $\alpha>M_A$,
\begin{gather*}
\no{(A-\alpha I)^{-1}}\leq (\alpha-M_A)^{-1}\,.
\end{gather*}
\end{corollary}

Let us turn our attention to the restriction \eqref{eq:restrict-relation}.

\begin{lemma}\label{lem:AvsAA-semi-bounded}
A closed symmetric relation $A$ is l.s.b. (u.s.b.) if and only if $A_A$ is l.s.b. (u.s.b.), with the same greatest lower (smallest upper) bound.
\end{lemma}
\begin{proof}
From \eqref{eq:Tt-ToP} and Remark~\ref{rmk:Asym-AA}, one has that $\dom A\subset (\mul A)^\perp$ and for $\vE fg\in A$ there exists $h\in\mul A$ such that $g=A_Af+h$. Thus, 
\begin{gather*}
\ip{f}{g}=\ip{f}{A_Af+h}=\ip f{A_Af}\,,
\end{gather*}which implies the assertion. 
\end{proof}

It is well-known that a selfadjoint operator $T$ in $\mathcal  H$ has a bounded spectrum if and only if it belongs to $\mathcal  {B(H)}$. In this case,
\begin{gather}\label{eq:bound-T-Selfadjoint}
\no T=\sup\lrb{\abs{\ip f{Tf}}\,:\, f\in \mathcal  H\,,\no f=1}\,.
\end{gather}

\begin{theorem}\label{th:Alusb-oP}
For a selfadjoint relation $A$, the following are equivalent:
\begin{enumerate}[(i)]
\item\label{it1-lusb} $A$ is both a lower and upper semi-bounded relation.
\item\label{it2-lusb} $\sigma(A)$ is a bounded subset of $\R$.
\item\label{it3-lusb} $A_A$ is a selfadjoint operator in $\mathcal  B((\mul A)^\perp)$.
\end{enumerate}
In such a case, the greatest lower and smallest upper bounds \eqref{eq:GS-bound} satisfy
\begin{gather}\label{eq:bound-AoP-Restrict}
\sigma(A)\subset[m_A,M_A]\quad\mbox{and}\quad\no{A_A}=\max\lrb{\abs{m_A},\abs{M_A}}\,.
\end{gather}
\end{theorem}
\begin{proof}
\eqref{it1-lusb}$\Rightarrow$\eqref{it2-lusb} If $A$ is both lower and upper semi-bounded, then Proposition~\ref{prop:characterization-sa}, Theorem~\ref{th:charact-lsb} and Corollary~\ref{cor:charact-usb} yield the left-hand side of \eqref{eq:bound-AoP-Restrict}, as required. \eqref{it2-lusb}$\Rightarrow$\eqref{it3-lusb} Since $A$ has a bounded spectrum, it follows from Remarks~\ref{rmk:Asym-AA} and \ref{rm:Aselfadjoint-restriction} that $A_A$ is a selfadjoint operator with bounded spectrum, viz. $A_A\in\mathcal  B((\mul A)^\perp)$. \eqref{it3-lusb}$\Rightarrow$\eqref{it1-lusb} One gets by \eqref{eq:bound-T-Selfadjoint} that $-\no{A_A}\no f^2\leq \ip{f}{A_A f}\leq \no {A_A}\no f^2$, for every $f\in\mathcal  H$. Hence, $A_A$ is both l.s.b. and u.s.b, and so is $A$, due to Lemma~\ref{lem:AvsAA-semi-bounded}. The right-hand side of \eqref{eq:bound-AoP-Restrict} follows by \eqref{eq:GS-bound}, \eqref{eq:bound-T-Selfadjoint} and Lemma~\ref{lem:AvsAA-semi-bounded}, bearing in mind that $-\inf\ip{f}{A_Af}=\sup-\ip{f}{A_Af}$.
\end{proof}

From Theorem~\ref{th:Alusb-oP}, the fact that a selfadjoint $A$ is both l.s.b. and u.s.b. means that its operator part is bounded, even if $A$ is not.

\section{The Krein transform for linear relations}\label{s4}

We shall work with the following version of the Krein transform for linear relations.
\begin{definition}
The \emph{Krein transform} of a linear relation $T$ is given by 
\begin{gather}\label{eq:def-KT}
\kT T\colonequals 2\lrp{T+I}^{-1}-I=\lrb{\vE{f+g}{f-g}\,:\,\vE fg\in T}\,,
\end{gather}
which is a linear relation with
\begin{equation}\label{eq:Kt-components}
\begin{aligned}
 \dom \kT T&=\ran (T+I)\,,& \ran \kT  T&=\ran (T-I)\,,
 \\ \ker \kT  T&=\ker (T-I)\,,&\mul \kT T&=\ker (T+ I)\,.
\end{aligned}
\end{equation}
\end{definition}

For relations $T$ and $S$, it is a simple matter to verify from \eqref{eq:def-KT} that the Krein transform satisfies the following properties:
\begin{equation}\label{eq:Kt-properties}
\begin{aligned}
& \kT {\kT T}=T\,,&\qquad  &T\subset S\ \Leftrightarrow\ \kT T\subset \kT S\,,\\
&\kT {T^{-1}}=-\kT T\,,&\qquad &\kT {T\dotplus S}=\kT T\dotplus \kT S\,,\\
&\kT {T}^{-1}=\kT {-T}\,,&\qquad&\kT{T\oplus S}={\kT T}\oplus{\kT S}\,.
\end{aligned}
\end{equation}

The first property of \eqref{eq:Kt-properties} means that the Krein transform is an involution.
\begin{proposition}\label{prop:kt-adjoint}
If $T$ is a linear relation, then the following holds:
\begin{enumerate}[(i)]
\item\label{it1-kt} $\kT {T^\perp}=\kT T^\perp$.
\item\label{it2-kt} $\kT{T^*}=\kT T^{*}$.
\item\label{it3-kt} $\cc{\kT T}=\kT {\cc T}$.
\end{enumerate}
\end{proposition}
\begin{proof}
Note that $\kT{T^\perp}\oplus\kT{T}=\kT{T^\perp\oplus T}\subset \HH$, i.e., $\kT{T^\perp}\subset\kT{T}^\perp$. In this fashion,
\begin{gather*}
T^\perp=\kT{\kT{T^\perp}}\subset\kT{\kT{T}^\perp}\subset \kT{\kT{T}}^\perp=T^\perp\,,
\end{gather*}
whence it follows that $\kT {T^\perp}=\kT T^\perp$. Besides, 
\begin{gather*}
\kT{T^*}=\kT{\lrp{-T^{-1}}^\perp}=\lrp{-\kT{T}^{-1}}^\perp=\kT{T}^*\,.
\end{gather*}
Furthermore, $\kT{\cc T}=\kT{T^{**}}=\kT{T}^{**}=\cc{\kT T}$, as required.
\end{proof}
\begin{remark}
It is evident from \eqref{eq:Kt-properties} and Proposition~\ref{prop:kt-adjoint} that a relation and its Krein transform are symmetric or selfadjoint simultaneously. Viz. the property of being a symmetric or selfadjoint relation is invariant under the Krein transform. 
\end{remark}

Let us deal with a particular class of semi-bounded relations.
\begin{definition}
We say that a relation $A$ is positive (briefly $A\geq0$) if 
\begin{gather}\label{eq:positive}
\ip fg\geq0\,,\quad \mbox{for all}\quad \vE fg\in A\,.
\end{gather}
\end{definition}

\begin{remark}
A positive relation is l.s.b., with greatest lower bound $m_A\geq0$. Moreover, if $A$ is l.s.b., then $A-m_A I\geq 0$.
\end{remark}

Recall that a relation $V$ is a contraction if 
\begin{gather}\label{eq:contraction}
\no k\leq \no h\,,\quad\mbox{for all}\quad\vE hk\in V\,.
\end{gather}
Besides, $V$ is isometric if the equality in \eqref{eq:contraction} holds, which is equivalent to $V^{-1}\subset V^*$.

\begin{theorem}\label{th:positive-Kt-contraction}
A linear relation $A$ is positive if and only if $\kT A$ is a symmetric contraction.
\end{theorem}
\begin{proof}
For $A\geq0$, it is clear that $\kT A$ is symmetric. Besides, for $\vE{f+g}{f-g}\in \kT A$  with $\vE fg\in A$, it follows that  $\ip fg\geq0$ and  
   \begin{align}\label{eq:positive-Kcontraction}
   \begin{split}
   \no {f-g}^2 &=\no g^{2}+\no f^{2}-2\ip fg\\
   &\leq \no g^{2}+\no f^{2}+2\ip fg= \no {f+g}^2\,.
   \end{split}
  \end{align}
Therefore, $\kT A$ is a contraction. Conversely, if  $\vE{h+k}{h-k}\in \kT {\kT A}=A$, with  $\vE hk\in \kT A$, then since $\kT A$ is a symmetric contraction, $\ip hk\in\R$,  $\no h\geq \no k$ and  
\begin{gather}\label{eq:contraction-Kpositive}
 \ip{h+k}{h-k}=\no h^2-\no k^2\geq 0\,.
\end{gather}
Hence, $A$ is positive.
\end{proof}
As a consequence of Theorem~\ref{th:positive-Kt-contraction}, the Krein transform gives a one-to-one correspondence between positive relations and symmetric contractions.

\begin{corollary}\label{coro:norm1-inposi}
If $A$ is a positive relation such that it is unbounded or has a nontrivial kernel, then $\kT A$ is a symmetric contraction with norm equal to one.  
\end{corollary}
\begin{proof}
By virtue of Theorem~\ref{th:positive-Kt-contraction}, $\kT A$ is symmetric with $\no{\kT A}\leq 1$. If $A$ is unbounded, then there exists $\lrb{\vE {f_n}{g_n}}_{n\in\N} \subset A$ such that $\no {f_n}\in\{0,1\}$ and $\no {g_n}\to\infty$. In this fashion,  $\vE {f_n+g_n}{f_n-g_n} \in\kT A$, for all $n\in\N$ and 
\begin{gather*}
 1 \geq\no {\kT A}\geq \frac{\no{f_n-g_n}}{\no{f_n+g_n}}\geq\frac{\no{g_n}-\no{f_n}}{\no{g_n}+\no{f_n}}\to1\,.
\end{gather*}
Therefore,  $\no {\kT A}=1$. For the case that there is a unit element $f\in\ker A$, one has that $\vE ff\in\kT A$, which readily implies  $\no {\kT A}=1$.
\end{proof}
A positive relation with a nontrivial multivalued part is unbounded and, hence, satisfies conditions of Corollary~\ref{coro:norm1-inposi}.  In this fashion, if a symmetric contraction $A$ satisfies $\no A<1$, then $\kT A$ is positive and bounded, with $\ker \kT A =\{0\}$.

The following belong to the class of positive relations.

\begin{definition}\label{def:quasi-null}
We call a relation $A$ \emph{quasi-null} if 
\begin{gather*}
\ip fg=0\,,\quad\mbox{for all\,} \vE fg\in A\,.
\end{gather*}
\end{definition}

A quasi-null relation $A$ is l.s.b. as well as u.s.b., with greatest lower and smallest upper bounds $m_A=M_A=0$ (see Definition~\ref{def:lsb-usb} and \eqref{eq:GS-bound}).

\begin{remark}\label{rm:spectral-core-qn}
It is a simple matter to verify from Theorem~\ref{th:charact-lsb} and Corollary~\ref{cor:charact-usb} that a relation $A$ is quasi-null if and only if it is symmetric, $\hat\sigma(A)\subset\{0\}$ and for all $\alpha\in\R\backslash\{0\}$, 
\begin{gather*}
\no{(A-\alpha I)^{-1}}\leq \abs\alpha^{-1}\,.
\end{gather*}
\end{remark}

\begin{corollary}\label{cor:qn-Kt-isometric}
A relation is quasi-null if and only if its Krein transform is both symmetric and isometric.
\end{corollary}
\begin{proof}
Following the same lines as the proof of Theorem~\ref{th:positive-Kt-contraction} and changing the inequalities~\eqref{eq:positive-Kcontraction} and \eqref{eq:contraction-Kpositive} by equalities, the assertion follows.
\end{proof}

Let us now present another characterization of quasi-null relations.
\begin{corollary}\label{cor:qn-equivalences} For a relation $A$, the following are equivalent:
\begin{enumerate}[(i)]
\item\label{it1-qrc} $A$ is quasi-null.
\item\label{it2-qrc} $A,-A\subset A^*$.
\item\label{it3-qrc} $\dom A\perp \ran A$. 
\end{enumerate}
\end{corollary}
\begin{proof}
\eqref{it1-qrc}$\Rightarrow$\eqref{it2-qrc} Since $A$ is quasi-null, then it is symmetric, i.e., $A\subset A^*$, and by Corollary~\ref{cor:qn-Kt-isometric}, $\kT A$ is isometric. Thus, by \eqref{eq:Kt-properties} and Proposition~\ref{prop:kt-adjoint}, 
\begin{gather*}
\kT{-A}=\kT A^{-1}\subset\kT A^*=\kT{A^*}\,,
\end{gather*}
whence $-A\subset A^*$. \eqref{it2-qrc}$\Rightarrow$\eqref{it3-qrc} For $\vE fg,\vE hk \in A$, it follows that  $\vE hk,\vE h{-k}\in A^*$. Thus,
$\ip fk=\ip gh$ and $-\ip fk=\ip gh$, which implies $\ip fk=0$. \eqref{it3-qrc}$\Rightarrow$\eqref{it1-qrc} It is straightforward.
\end{proof}

The quasi-null relations will be practical in Section~\ref{s5} to find semi-bounded extensions of linear relations.

\section{Semi-bounded extensions of semi-bounded linear relations}\label{s5}
We describe in this section the semi-bounded extensions of a semi-bounded relation $A$. In some cases, we restrict our discussion to symmetric extensions of $A$ since we are focused on characterizing only those that are semi-bounded (see \cite[Sect.\,4]{MR4082306} for symmetric extensions of linear relations).

If $A$ is a closed symmetric relation for which $\hat\rho(A)\cap\R\neq\emptyset$. In particular, if $A$ is a closed semi-bounded relation (see Theorem~\ref{th:charact-lsb} and Corollary~\ref{cor:charact-usb}), then $\hat\rho(A)$ consists of one component only, and its indices \eqref{eq:deficiency-index} are equal. In such a case, we denote the index of $A$ by 
\begin{gather}\label{eq:semi-bounded-index}
\eta_A\colonequals \dim \dS_{\zeta}(A^*)\,,\quad \zeta\in\hat\rho(A)
\end{gather}

\begin{remark}\label{rm:sb-deficiency-index}
For a  closed semi-bounded relation $A$, the following hold:
\begin{enumerate}[(i)]
\item\label{it0:sb-di} Every closed symmetric extension $S$ of $A$ satisfies (cf. \cite[Cor.\,4.8]{MR4082306})
\begin{gather}\label{eq:indices-symmetric-extensions}
\eta_A=\dim \dS_{\pm i}(S^*)+\dim[S/ A]\,.
\end{gather}
\item\label{it1:sb-di} Proposition~\ref{prop:characterization-sa} and \eqref{eq:indices-symmetric-extensions} imply that $A$ is either selfadjoint or has closed symmetric extensions with selfadjoint extensions. 
\item\label{it2:sb-di} For $\eta_A<\infty$, all the selfadjoint extensions $S$ of $A$ have the same essential spectrum (cf. \cite[Sect.\,4]{MR4091412})
\begin{gather*}
\sigma_e(S)\colonequals\sigma_p^\infty(S)\cup\sigma_c(S)\,. 
\end{gather*}
\end{enumerate}
\end{remark}

It readily follows that the bound \eqref{eq:GS-bound} of any l.s.b. extension $S$ of a l.s.b. relation $A$ satisfies $m_A\geq m_S$. Similarly, $M_A\leq M_S$ holds for any u.s.b. extension $S$ of an u.s.b. relation $A$.

\begin{proposition}\label{prop:existence-sbe}
A l.s.b. (resp. u.s.b) relation has l.s.b. (resp. u.s.b) selfadjoint extensions.
\end{proposition}
\begin{proof}
If $A$ is a l.s.b. relation with greatest lower bound $m_A$, then (cf. \cite{MR500265})  there exist two positive selfadjoint extensions $S_N,S_F$ of $A-m_AI\geq0$. Hence, one concludes from Remark~\ref{rm:dist-adjoint} that $S_N+m_AI$ and $S_F+m_AI$ are l.s.b. selfadjoint extensions of $A$. The proof when $A$ is u.s.b. follows the same lines as above.
\end{proof}

The following assertion will be practical in the sequel and, for the reader's convenience, is adapted from \cite[Props.\, 4.10 and 4.11]{MR4082306}.
\begin{lemma}\label{lem:sa-extension-lacunas}
Let $A$ be a closed symmetric relation such that $\hat\rho(A)\cap\R\neq\emptyset$. If $\eta_A<\infty$ then for every $\alpha\in\hat\rho(A)\cap\R$, the relation
\begin{gather*}
S_\alpha=A\dotplus\dS_\alpha(A^*)
\end{gather*}
is the unique selfadjoint extension of $A$ such that $\dim\dS_\alpha(S_\alpha)=\eta_A$, viz. $\alpha$ is an eigenvalue of multiplicity $\eta_A$. Besides, $\sigma(S_\alpha)\cap\hat\rho(A)\subset\sigma_d(S_\alpha)$, where the eigenvalues are of multiplicity at most $\eta_A$.
\end{lemma}

Proposition~\ref{prop:existence-sbe} asserts the existence of semi-bounded selfadjoint extensions of a semi-bounded relation $A$. In what follows, we shall explicitly give semi-bounded extensions of $A$.

\begin{theorem}\label{th:sa-extension-lusb}
Let $A$ be a closed  l.s.b. (resp. u.s.b.) with index $\eta_A<\infty$. Then for $\alpha<m_A$ (resp. $\alpha>M_A$), the relation $S_\alpha=A\dotplus\dS_\alpha(A^*)$ satisfies the following:
\begin{enumerate}[(i)]
\item\label{it1-sa-lusb} $S_\alpha$ is the unique l.s.b. (resp. u.s.b.) selfadjoint extension of $A$, with greatest lower bound $m_{S_\alpha}=\alpha$ (resp. smallest upper bound $M_{S_\alpha}=\alpha$). 
\item\label{it2-sa-lusb} $\sigma(S_\alpha)\cap[\alpha, m_A)$ (resp. $\sigma(S_\alpha)\cap(M_A,\alpha]$) consist solely of isolated eigenvalues of multiplicity at most $\eta_A$.  In particular, $\alpha\in\sigma_p(S_\alpha)$ is of multiplicity $\eta_A$.
\item\label{it3-sa-lusb} $\sigma_e(S_{\alpha_1})=\sigma_e(S_{\alpha_2})$, for every $\alpha_1,\alpha_2<m_A$ (resp. $\alpha_1,\alpha_2>M_A$).
\end{enumerate}
\end{theorem}
\begin{proof}
Since $A$ is closed and l.s.b.,  $\eta_A<\infty$ and $\alpha<m_A$, it follows from Theorem~\ref{th:charact-lsb} and Lemma~\ref{lem:sa-extension-lacunas} that $S_\alpha=A\dotplus\dS_\alpha(A^*)$ is the unique selfadjoint extension of $A$, such that $\alpha\in\sigma_p(S_\alpha)$ and it is of multiplicity $\eta_A$. Now, for every $\vE{f+h}{g+\alpha h}\in S_\alpha$, with $\vE fg\in A$ and $\vE h{\alpha h}\in A^*$, one has that $\ip fg\geq m_A\no f^2$ and $\ip hg=\alpha\ip hf$. Thus, 
\begin{align*}
\ip{f+h}{g+\alpha h}\geq m_A\no f^2+\alpha 2{\rm Re\,}\ip fh+\alpha\no h^2>\alpha\no{f+h}^2\,,
\end{align*}
whence it follows that $S_\alpha$ is l.s.b. and $m_{S_\alpha}\geq\alpha$, which are equals, otherwise $\alpha\in\rho(S_\alpha)$, a contradiction since $\alpha\in\sigma_p(S_\alpha)$. This proves item \eqref{it1-sa-lusb} and item \eqref{it2-sa-lusb} follows from Theorem~\ref{th:charact-lsb} and Lemma~\ref{lem:sa-extension-lacunas}, while item \eqref{it3-sa-lusb} follows from Remark~\ref{rm:sb-deficiency-index}-\eqref{it2:sb-di}. The case when $A$ is u.s.b is simple, since $-A$ is l.s.b. 
\end{proof}

We now turn our attention to positive and quasi-null relations.

\begin{corollary} If $A$ is a closed positive relation with index $\eta_A<\infty$. Then
\begin{gather*}
S_\alpha=A\dotplus\dS_\alpha(A^*)\,,\quad \alpha<0
\end{gather*}
is the unique l.s.b. selfadjoint extension of $A$, with greatest lower bound $m_{S_\alpha}=\alpha$. Besides, $\sigma(S_\alpha)\cap[\alpha, 0)$ consist solely of isolated eigenvalues of multiplicity at most $\eta_A$. In particular, $\alpha\in\sigma_p(S_\alpha)$ is of multiplicity $\eta_A$. Moreover, $\sigma_e(S_{\alpha_1})=\sigma_e(S_{\alpha_2})$, for every $\alpha_1,\alpha_2<0$.
\end{corollary}
\begin{proof}
Setting $m_A=0$ in Theorem~\ref{th:sa-extension-lusb}, one yields the assertion.
\end{proof}

The following shows a similar result to the previous one for quasi-null relations, whose proof is straightforward from Theorem~\ref{th:sa-extension-lusb}.

\begin{corollary}\label{coro:sb-extensions-qn-fi}
If $A$ is a closed quasi-null relation, with index $\eta_A<\infty$. Then, every $\alpha\in\R\backslash\{0\}$ is an eigenvalue of multiplicity $\eta_A$, of one, and only one, selfadjoint extension of $A$, given by
\begin{gather}\label{eq:unique-SA-sb}
S_\alpha=A\dotplus\dS_\alpha(A^*)\,.
\end{gather}
Besides,  $\sigma(S_\alpha)\cap\R\backslash\{0\}$ consists of isolated eigenvalues of multiplicity at most $\eta_A$ and, therefore, $\sigma_e(S_\alpha)\subset \{0\}$. All the selfadjoint extensions \eqref{eq:unique-SA-sb} have the same essential spectrum, and
when $\alpha<0$, $S_\alpha$ is l.s.b., while for $\alpha>0$, $S_\alpha$ is u.s.b.
\end{corollary}

\begin{remark}
By virtue of Remark~\ref{rm:sb-deficiency-index}-\eqref{it1:sb-di}, a closed positive relation $A$ has closed symmetric extensions, which can be determined by the so-called second von Neumann formula for linear relations. Namely, $S$ is a closed symmetric extension of $A$ if and only if 
 \begin{gather}\label{eq:von02d}
  S=A\oplus (V- I)D\,,
 \end{gather}
 where $D\subset \dS_i (A^*)$ is a closed bounded relation, and $ V:D\to \dS_{-i} (A^*)$ is a closed isometry in $(\HH)\oplus(\HH)$ \cite[Thm.\,4.7]{MR4082306}.
\end{remark}

Let us restrict our discussions to finding positive extensions of positive relations.

\begin{theorem}\label{th:positive-extensions}
Let $A$ be a closed positive relation and $S$ a closed symmetric extension of $A$. Then $S$ is positive if and only if there exists a unique closed positive relation $L\subset A^*$ such that 
\begin{gather}\label{eq:positive-extension}
S=A\oplus L\,,
\end{gather}
and, for every $\vE fg\in A$ and $\vE hk\in L$, 
\begin{gather}\label{eq:positive-extension-condition}
\ip fg+\ip hk\geq 2\abs{a}\,,\quad a\in\{{\rm Re\,}\ip fk,\,{\rm Im\,}\ip fk\}.
\end{gather}
When $A$ is quasi-null, the condition \eqref{eq:positive-extension-condition} turns into $\ip fk=\ip gh=0$.
\end{theorem}
\begin{proof}
We first suppose that $S$ is a closed positive extension of $A$. Then, it is easy to verify that $L=S\ominus A\subset A^*$ and it is a closed positive relation, which satisfies \eqref{eq:positive-extension}. Besides, for $\vE fg\in A\subset S,\vE hk\in L\subset S$, it follows that $\vE{\alpha f+h}{\alpha g+k}\in S$, with $\alpha\in\C$. Since $S$ is positive, $\ip kf=\ip hg$ and 
\begin{gather}\label{eq:aux-eq-dS}
0\leq\ip{\alpha f+h}{\alpha g+k}=\abs{\alpha}^2\ip fg+\ip hk+2{\rm Re\,}(\alpha\ip kf)\,.
\end{gather}
Hence, taking $\alpha\in\lrb{\pm 1,\pm i}$, one arrives at \eqref{eq:positive-extension-condition}. If $A$ is quasi-null then $\ip fg=0$, and if $\ip fk\neq0$, there would exist $\beta>0$ such that $0>\ip hk-\beta\abs{\ip fk}$, a contradiction with \eqref{eq:aux-eq-dS}, for $\alpha=-\beta\abs{\ip fk}/\ip kf$. Hence, $\ip gh=\ip fk=0$. 

The uniqueness is simple since if $\hat L$ satisfies the conditions of $L$. Thus, one has at once that $\hat L=L$. The converse assertion is straightforward.
\end{proof}

\begin{remark}
The uniqueness of Theorem~\ref{th:positive-extensions} implies that the relation $L$ in \eqref{eq:positive-extension} coincides with the relation $(V-I)D$ in \eqref{eq:von02d}.
\end{remark}

The following gives a characterization of quasi-null relations through its positive extensions.
\begin{corollary}
Let $S$ be a closed positive extension of a closed positive relation $A$. Then, $A$ is quasi-null if and only if $-S\subset A^*$.
\end{corollary}
\begin{proof}
If $A$ is quasi-null then by Corollary~\ref{cor:qn-equivalences} $-A\subset A^*$. Besides, the decomposition \eqref{eq:positive-extension} asserts that for every $\vE fg\in A$ and $\vE hk\in L$,  $\ip{f}{-k}=0=\ip{g}{h}$, which implies $-L\subset A^*$. Hence, $-S=(-A)\oplus(-L)\subset A^*$. Conversely, since $A$ is positive and $A\subset S$, then $A\subset A^*$ and  $-A\subset -S\subset A^*$.  Therefore, one has by Corollary~\ref{cor:qn-equivalences} that $A$ is quasi-null.
\end{proof}

We conclude this section by presenting the analogous formula of \eqref{eq:von02d} for positive extensions of quasi-null relations.
\begin{theorem}\label{th:positive-extensions-qn}
Let $A$ be a closed quasi-null relation and $S$ a closed symmetric extension of $A$. Then $S$ is positive (quasi-null) if and only if
\begin{gather}\label{eq:positive-extension-qn}
S=A\oplus(V-I)D\,,
\end{gather}
where $D\subset\dS_{1}(-A^*)$ is a closed bounded relation, and $V\colon D\to \dS_{-1}(-A^*)$ is a closed contraction (isometry) in $(\HH)\oplus(\HH)$. 
\end{theorem}
\begin{proof}
If $S$ is positive (quasi-null),  then one has from Theorem~\ref{th:positive-extensions} that $S=A\oplus L$, where $L\subset A^*$ is a closed positive (quasi-null at once with $S$) relation such that 
\begin{gather}\label{eq:orthogonal-c}
\ip gh=\ip fk=0\,,\quad \vE fg\in A,\vE hk\in L\,.
\end{gather}
Now, Theorem~\ref{th:positive-Kt-contraction} (Corollary~\ref{cor:qn-Kt-isometric}) implies that $W=\kT L$ is a closed symmetric contraction (isometry). For $u\in\dom W$, it follows by \eqref{eq:Kt-components} that $u=k+h$ and by \eqref{eq:orthogonal-c},
\begin{gather*}
\ip{f+g}{k+h}=\ip fh+\ip gk=\ip{\vE fg}{\vE hk}=0\,,
\end{gather*}
which implies $\dom W\perp \ran(A+I)$. Similarly, $\ran W\perp \ran(A-I)$. So, by \eqref{eq:Hdecomposition},
\begin{align*}
\dom W\subset \ker(A^*+I)\,,\qquad \ran W\subset \ker(A^*-I)\,.
\end{align*}
Note that $\vE uu\in\dS_1(-A^*)$ and $\vE v{-v}\in\dS_{-1}(-A^*)$, for $\vE uv\in W$. Consider 
\begin{gather*}
D=\lrb{\vE uu\,:\, u\in\dom W}\subset\dS_1(-A^*)
\end{gather*}
and $V\colon D\to \dS_{-1}(-A^*)$ given by $V\vE uu=\vE {-v}v$. Since $\dom W$ is closed, $D$ is closed and bounded. Moreover, it is clear that $V$ is a closed contraction (isometry) in $(\HH)\oplus(\HH)$, since $W$ is. Hence, 
\begin{gather*}
L=\kT W=\lrb{\vE{u+v}{u-v}\,:\,\vE uv\in W}=\lrb{\lrp{I-V}\vE {u}{u}\,:\,\vE uu\in D}=(V-I)D
\end{gather*}
whence one arrives at \eqref{eq:positive-extension-qn}. We have used $(V-I)D=(I-V)D$.

Conversely, denote $L=(V-I)D\subset -A^*$, which is a closed symmetric relation such that it is contained in $A^*$, since $S\subset A^*$. Thus, for  $\vE fg\in A,\vE hk\in L$,
\begin{gather}\label{eq:l-null-condition}
\ip fk=\ip gh\,,\quad -\ip fk=\ip gh\quad\mbox{whence}\quad \ip fk=0\,.
\end{gather}
Besides, for $\vE hk=\vE{u+v}{u-v}\in L$, with $\vE uu\in D$ and $\vE {-v}v=V\vE uu$, since $V$ is a contraction,
\begin{gather}\label{eq:L-positive-qn-aux}
\ip hk=\ip{u+v}{u-v}=\frac12\lrp{\no{\vE uu}^2-\no{\vE v{-v}^2}}\geq 0\,,
\end{gather}
i.e., $L\geq0$. Hence, by virtue of Theorem~\ref{th:positive-extensions}, $S$ is a closed positive extension of $A$. When $V$ is an isometry, one has by \eqref{eq:L-positive-qn-aux} that $L$ is quasi-null. So, in addition to \eqref{eq:l-null-condition}, one readily shows that $S$ is also quasi-null. 
 \end{proof}

\section{Examples: semi-bounded extensions of the infinite star operator}\label{s6}

For a countable, simple and connected graph $G$, with a set of vertices 
\begin{gather}\label{eq:set-vertices}
V\colonequals\{\delta_j\}_{j\in\N_0}\,,\qquad \lrp{\N_0=\N\cup\{0\}} 
\end{gather}
we consider the Hilbert space of square-summable sequences  $(l_2(V),\ip{\cdot}{\cdot})$, having \eqref{eq:set-vertices} as canonical bases. The adjacency operator $A$ of $G$ is the linear operator which maps every vertex $\delta_j\in \dom A$ into the sum of its adjacent vertices, i.e., 
\begin{gather*}
A\delta_j=\sum_{\delta_k\sim\delta_j}\delta_k\,.
\end{gather*}

\subsection{The infinite star graph}\label{sub:isg}

Let us consider $G$ as the infinite star graph (see Fig.\, \ref{fig:graph1}). Then, for $\mathds 1=\sum_{j\in\N}\delta_j$, the adjacency operator of $G$ is a nondensely defined operator given by
\begin{gather*}
A=\lrb{\vE{f}{\ip{\mathds 1}{f}\delta_0}\,:\, f\in{l_2(V)}\ominus{\{\delta_0\}}\,,\, \abs{\ip{\mathds 1}{f}}<\infty}\,.
\end{gather*}
In this fashion, $G$ turns to be directed.

\begin{figure}[h]
\centering
\begin{tikzpicture}[>=stealth',shorten >=1pt,node distance=3cm,on grid,initial/.style    ={}]
  [scale=.5,auto=left,every node/.style={}]
  \node (n0) at (4.5,8) {$\delta_0$};
  \node (n2) at (1.5,7)  {$\delta_1$};
  \node (n3) at (3,6)  {$\delta_2$};
  \node (n4) at (6,6) {$\delta_3$};
  \node (n5) at (7.5,7)  {$\iddots$};
\tikzset{mystyle/.style={-,double=black}} 
\tikzset{every node/.style={fill=white}} 
\path (n0)     edge [mystyle]    (n2);  
\path (n0)     edge [mystyle]    (n3);  
\path (n0)     edge [mystyle]    (n4);    
\path (n0)     edge [mystyle]    (n5);      
\end{tikzpicture}
\hskip.5cm
\begin{tikzpicture}[>=stealth',shorten >=1pt,node distance=3cm,on grid,initial/.style    ={}]
  [scale=.5,auto=left,every node/.style={}]
  \node (n0) at (4.5,8) {$\delta_0$};
  \node (n2) at (1.5,7)  {$\delta_1$};
  \node (n3) at (3,6)  {$\delta_2$};
  \node (n4) at (6,6) {$\delta_3$};
  \node (n5) at (7.5,7)  {$\iddots$};
\tikzset{mystyle/.style={<-,double=black}} 
\tikzset{every node/.style={fill=white}} 
\path (n0)     edge [mystyle]    (n2);  
\path (n0)     edge [mystyle]    (n3);  
\path (n0)     edge [mystyle]    (n4);    
\path (n0)     edge [mystyle]    (n5);      
\end{tikzpicture}
\caption{The infinte star graph turns to be directed.}\label{fig:graph1}
\end{figure}
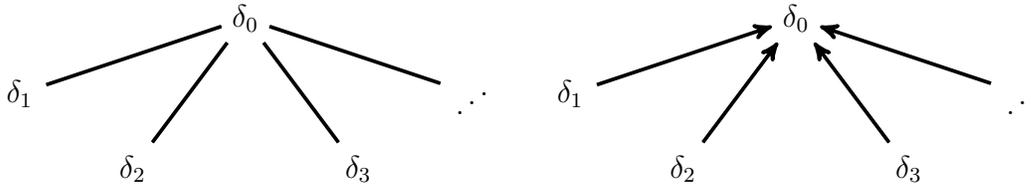

It is clear that $A$ is quasi-null. Besides, one easily computes that the linear functional $ f\in\dom A\mapsto \ip{\mathds 1}{f}$ is not continuous. So, $A$ is not bounded nor closable (cf. \cite[Subsect.\,1.1.2]{MR2953553}).

\begin{proposition}\label{prop:A-closure}
The closure of $A$ is the quasi-null selfadjoint relation
\begin{gather*}
S=0\upharpoonleft_{l_2(V)\ominus\{\delta_0\}}\oplus\Span\lrb{\vE0{\delta_0}}\,.
\end{gather*}
Moreover, $\sigma_p^\infty(S)=\sigma(S)=\{0\}$ and $\sigma_c(S)=\emptyset$.
\end{proposition}
\begin{proof}
The relation $S$ is closed, quasi-null, and contains $\cc A$. Note that $S_S=0$, which is selfadjoint in $l_2(V)\ominus\{\delta_0\}$, and so is $S$, by \eqref{eq:espectral-properties-inherited}, Remark~\ref{rmk:Asym-AA} and Proposition~\ref{prop:characterization-sa}. Now, if 
\begin{gather*}
\vE hk\in A^*\,,\quad \mbox{with}\quad h=\sum_{j\in \N_0}h_j\delta_j, k=\sum_{j\in \N_0}k_j\delta_j\in l_2(V)\,,
\end{gather*}
then $k_n=\ip{\delta_n}k=\ip{\delta_0}h=h_0$, for all $j\in\N$, since $A\delta_n=\delta_0$. It follows that $h_0=0$ and $k=k_0\delta_0$ and $A^*\subset S$. Thus, $S\subset A^{**}=\cc A$ and, hence, equals. Observe that \eqref{eq:espectral-properties-inherited} and Remark~\ref{rmk:Asym-AA} imply  $\sigma_p^\infty(S)=\sigma(S)=\{0\}$.  The above also implies $\sigma_c(S)\subset\{0\}$, and $\sigma_c(S)=\emptyset$ since $\ran S$ is finite-dimensional. 
\end{proof}

\subsection{The infinite directed and weighted star graph} We now consider the directed star graph of Subsection~\ref{sub:isg}, with the condition that its edges are weighted by $w=\sum_{j\in\N}w_j\delta_j\in l_2(V)\ominus\{\delta_0\}$ (see Fig.~\ref{fig:graph2}), where $w_j$ can be complex numbers, but for simplicity of notation, we choose $w_j\in\R\backslash\{0\}$. 
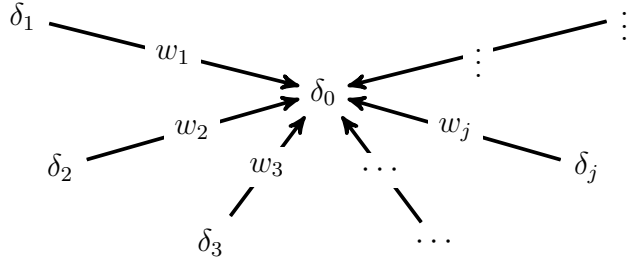
\begin{figure}[h]
\centering
\begin{tikzpicture}[>=stealth',shorten >=1pt,node distance=3cm,on grid,initial/.style    ={}]
  [scale=.5,auto=left,every node/.style={}]
  \node (n0) at (4.5,8) {$\delta_0$};
  \node (n1) at (0.5,9)  {$\delta_1$};
  \node (n2) at (1,7)  {$\delta_2$};
  \node (n3) at (3,6)  {$\delta_3$};
  \node (n4) at (6,6) {$\cdots$};
  \node (n5) at (8,7)  {$\delta_j$};
    \node (n6) at (8.5,9)  {$\vdots$};
\tikzset{mystyle/.style={<-,double=black}} 
\tikzset{every node/.style={fill=white}} 
\path (n0)     edge [mystyle]    node   {$w_1$}  (n1) ;   
\path (n0)     edge [mystyle]    node   {$w_2$} (n2);  
\path (n0)     edge [mystyle]    node   {$w_3$} (n3);  
\path (n0)     edge [mystyle]    node   {$\cdots$} (n4);    
\path (n0)     edge [mystyle]    node   {$w_j$} (n5);    
\path (n0)     edge [mystyle]    node   {$\vdots$} (n6);    
\end{tikzpicture}
\caption{The infinite directed and weighted star graph.}\label{fig:graph2}
\end{figure}

In this fashion, the adjacency operator of the graph is
\begin{gather}\label{eq:adjacency-ws}
A=\lrb{\vE f{\ip{w}{f}\delta_0}\,:\,f\in{l_2(V)}\ominus\{\delta_0\}}\,.
\end{gather}
which is closed, bounded and quasi-null, with $\dom A={l_2(V)}\ominus\{\delta_0\}$.

\begin{proposition}
The operator \eqref{eq:adjacency-ws} satisfies
\begin{gather}\label{eq:sepctral-core-ws}
\sigma_p^\infty(A)=\hat\sigma(A)=\{0\}\quad\mbox{and}\quad \sigma_c(A)=\emptyset\,.
\end{gather}
\end{proposition}
\begin{proof} Note that $A(w_1^{-1}\delta_1-w_j^{-1}\delta_j)=0$, for $j=2,3,\dots$ Hence, by Remark~\ref{rm:spectral-core-qn},  
\begin{gather*}
\{0\}\subset \sigma_p^\infty(A)\subset \hat\sigma(A)\subset\{0\}\,.
\end{gather*}
Besides, since $\ran A=\Span\{\delta_0\}$, then $0\notin\sigma_c(A)$. Thereby, one obtains \eqref{eq:sepctral-core-ws}.
\end{proof}

The following characterized the adjoint of \eqref{eq:adjacency-ws} and its deficiency space. 
\begin{lemma} The adjoint of  \eqref{eq:adjacency-ws} is given by 
\begin{gather}\label{eq:adjoint-iwsg}
A^*=\lrb{\vE{h}{\ip{\delta_0}hw}\,:\, h\in {l_2(V)}}\oplus\Span\lrb {\vE 0{\delta_0}}\,.
\end{gather}
Besides, 
\begin{gather}\label{eq:deficiency-index-iwsg}
 \dS_\zeta(A^*)=\Span\lrb{\vE{\delta_0+w/\zeta}{\zeta\delta_0+w}}\,,\quad \zeta\in\C\backslash\{0\}\,.
\end{gather}
Hence, the index \eqref{eq:semi-bounded-index} of $A$ is $\eta_A=1$.
\end{lemma}
\begin{proof} If we denote the right-hand side of \eqref{eq:adjoint-iwsg} by $S$, then a simple computation shows that $S\subset A^*$. Besides, since $A\delta_j=w_j\delta_0$,  for $j\in\N$, if $\vE hk\in A^*$, with $k=\sum_{l\in\N_0}k_l\delta_l$,  then  $k_j=\ip{\delta_j}{k}=w_j\ip{\delta_0}{h}$, which implies $A^*\subset S$, i.e., $A^*=S$. 

Now, for $\zeta\in\C\backslash\{0\}$ and $u=\delta_0+w/\zeta$, we check at once that $\Span\lrb{\vE{u}{\zeta u}}\subset A^*$. To show the other inclusion, if $\vE v{\zeta v}\in A^*$, then $\ip{\delta_0}v\neq0$, otherwise $v=0$. Thereby, we may suppose that $\ip{\delta_0}v=1$ and by \eqref{eq:adjoint-iwsg}, $\zeta v=w+\zeta\delta_0$, which implies $v=u$, viz. \eqref{eq:deficiency-index-iwsg}.
\end{proof}

\begin{remark} All the selfadjoint extensions of \eqref{eq:adjacency-ws} are in one-to-one correspondence with $\beta\in\{\zeta\in\C\,:\,\abs\zeta=1\}$, given by
\begin{gather}\label{eq:sa-extensions-A}
S_\beta=A\oplus\Span\lrb{\vE{i(\beta+1)w+(\beta-1)\delta_0}{(\beta-1)w-i(\beta+1)\delta_0}}\,.
\end{gather}
Indeed, since $\eta_A=1$, one has from Remark~\ref{rm:sb-deficiency-index} that all the symmetric extensions of $A$ are selfadjoint and given by \eqref{eq:von02d}. Hence, by \eqref{eq:deficiency-index-iwsg}, one arrives at \eqref{eq:sa-extensions-A}.
\end{remark}

\begin{theorem} There exists only one positive selfadjoint extension of \eqref{eq:adjacency-ws}, which is quasi-null, given by 
\begin{gather}\label{eq:sa-extension-qn-A}
S_1=A\oplus\Span\lrb{\vE{w}{\delta_0}}=0\upharpoonleft_{l_2(V)\ominus\{\delta_0\}}\oplus\Span\lrb{\vE0{\delta_0}}\,,
\end{gather}
and satisfies $\sigma_p^\infty(S_1)=\sigma(S_1)=\{0\}$, $\sigma_c(S_1)=\emptyset$. Besides, for every $\alpha\in\R\backslash\{0\}$,
\begin{gather}\label{eq:lsb-usb-extensions}
A_\alpha=A\dotplus\Span\lrb{\vE{\delta_0+w/\alpha}{\alpha\delta_0+w}}
\end{gather}
is the unique selfadjoint extension of $A$, with eigenvalues $\sigma_d(A_\alpha)=\{\alpha,-\no w^2/\alpha\}$ of multiplicity one, $\sigma_p^\infty(A_\alpha)=\{0\}$ and $\sigma_c(A_\alpha)=\emptyset$. Hence, 
\begin{gather*}
\sigma(A_\alpha)=\{0,\alpha,-\no w^2/\alpha\}\,.
\end{gather*}
\end{theorem}
\begin{proof} According to Theorem~\ref{th:positive-extensions}, the selfadjoint extension \eqref{eq:sa-extensions-A} is positive if 
\begin{gather*}
0=\ip w{(\beta-1)w-i(\beta+1)\delta_0}=(\beta-1)\no w^2\,,
\end{gather*}
i.e., when $\beta=1$. The spectral properties of \eqref{eq:sa-extension-qn-A} follows from Proposition~\ref{prop:A-closure}. For the second part, Corollary~\ref{coro:sb-extensions-qn-fi} and \eqref{eq:deficiency-index-iwsg} imply \eqref{eq:lsb-usb-extensions}, and its eigenvalues are of multiplicity one, with $\alpha\in\sigma_d(A_\alpha)$. Besides, one obtains from Remark~\ref{rm:sb-deficiency-index}-\eqref{it2:sb-di} that $\sigma_p^\infty(A_\alpha)=\{0\}$ and $\sigma_c(A_\alpha)=\emptyset$. Now, if $\vE u{\zeta u}\in A_\alpha$, with $\zeta\in\R$ and $\zeta\notin\{0,\alpha\}$, then by \eqref{eq:lsb-usb-extensions}, there exist $f\in\dom A$, $a\in\C$, such that $u=f+a(\delta_0+w/\alpha)$ and 
\begin{gather*}
\zeta\lrp{f+a(\delta_0+w/\alpha)}=\ip wf\delta_0+a(\alpha\delta_0+w)\,,
\end{gather*}
which implies $\zeta f=-a(\zeta-\alpha)w/\alpha$ and $\ip wf=a(\zeta-\alpha)$. Since $a\neq0$ (otherwise $0=f=u$), we may suppose $a=1$. Thereby, 
\begin{gather*}
\zeta (\zeta-\alpha)=\ip w{\zeta f}=-(\zeta-\alpha)\frac{\no w^2}\alpha\,,
\end{gather*}
whence one deduces $\zeta=-\no w^2/\alpha$, with eigenvector $u=\delta_0-\alpha w/\no w^2$.\end{proof}

\begin{remark}
 On account of Theorem~\ref{th:positive-extensions-qn}, the quasi-null extension \eqref{eq:sa-extension-qn-A} can be described by 
\begin{gather*}
S_1=A\oplus (V-I)\dS_1(-A^*)\,,
\end{gather*}
where $V\colon \dS_1(-A^*)\to \dS_{-1}(-A^*)$, is such that $V\vE{\delta_0-w}{-\delta_0+w}=\vE{\delta_0+w}{\delta_0+w}$.
\end{remark}

\subsection*{Acknowledgments} This work was partially supported by CONACYT-Mexico (Grant CF-2019-684340) and  UAM-PEAPDI 2023: ``Semigrupos cu\'anticos de Markov: Operadores de transici\'on de niveles de energ\'ia y sus generalizaciones".

\def\cprime{$'$} \def\lfhook#1{\setbox0=\hbox{#1}{\ooalign{\hidewidth
  \lower1.5ex\hbox{'}\hidewidth\crcr\unhbox0}}} \def\cprime{$'$}
  \def\cprime{$'$} \def\cprime{$'$} \def\cprime{$'$} \def\cprime{$'$}
  \def\cprime{$'$} \def\cprime{$'$}
\providecommand{\bysame}{\leavevmode\hbox to3em{\hrulefill}\thinspace}
\providecommand{\MR}{\relax\ifhmode\unskip\space\fi MR }
\providecommand{\MRhref}[2]{%
  \href{http://www.ams.org/mathscinet-getitem?mr=#1}{#2}
}
\providecommand{\href}[2]{#2}

\end{document}